\documentclass[12pt,letterpaper]{article}
\usepackage{amsmath,amsfonts,amsthm,amssymb}

\theoremstyle{plain}

\numberwithin{equation}{section}
\newtheorem{thm}{Theorem}[section]
\newtheorem{lem}[thm]{Lemma}
\newtheorem{cor}[thm]{Corollary}

\newcounter{cond}

\newcommand{\complex}{{\mathbb C}}
\newcommand{\positive}{{\mathbb N}}
\newcommand{\real}{{\mathbb R}}
\newcommand{\ascript}{{\mathcal A}}
\newcommand{\bscript}{{\mathcal B}}
\newcommand{\cscript}{{\mathcal C}}
\newcommand{\mscript}{{\mathcal M}}
\newcommand{\qscript}{{\mathcal Q}}
\newcommand{\sscript}{{\mathcal S}}
\newcommand{\rmre}{\mathrm{Re\,}}
\newcommand{\rmtr}{\mathrm{tr}}
\newcommand{\rmspan}{\mathrm{span}}
\newcommand{\rmnorm}{\mathrm{norm}}
\newcommand{\rmrange}{\mathrm{Range}}
\newcommand{\rmcyl}{\mathrm{cyl}}
\newcommand{\dhat}{\widehat{D}}
\newcommand{\fhat}{\widehat{f}}
\newcommand{\ghat}{\widehat{g}}
\newcommand{\hhat}{\widehat{h}}
\newcommand{\muhat}{\widehat{\mu}}
\newcommand{\chihat}{\widehat{\chi}}
\newcommand{\gammahat}{\widehat{\gamma}}
\newcommand{\mutilde}{\widetilde{\mu}}
\newcommand{\cbar}{\bar{c}}
\newcommand{\fbar}{\bar{f}}

\newcommand{\ab}[1]{\left|#1\right|}
\newcommand{\doubleab}[1]{\left\|#1\right\|}
\newcommand{\brac}[1]{\left\{#1\right\}}
\newcommand{\paren}[1]{\left(#1\right)}
\newcommand{\sqbrac}[1]{\left[#1\right]}
\newcommand{\elbows}[1]{{\left\langle#1\right\rangle}}
\newcommand{\ket}[1]{{\left|#1\right>}}
\newcommand{\bra}[1]{{\left<#1\right|}}

\begin{document}

\title{DISCRETE QUANTUM PROCESSES
}
\author{S. Gudder\\ Department of Mathematics\\
University of Denver\\ Denver, Colorado 80208, U.S.A.\\
sgudder@.du.edu
}
\date{}
\maketitle

\begin{abstract}
A discrete quantum process is defined as a sequence of local states $\rho _t$, $t=0,1,2,\ldots$, satisfying certain conditions on an $L_2$ Hilbert space $H$. If $\rho =\lim\rho _t$ exists, then $\rho$ is called a global state for the system. In important cases, the global state does not exist and we must then work with the local states. In a natural way, the local states generate a sequence of quantum measures which in turn define a single quantum measure
$\mu$ on the algebra of cylinder sets $\cscript$. We consider the problem of extending $\mu$ to other physically relevant sets in a systematic way. To this end we show that $\mu$ can be properly extended to a quantum measure $\mutilde$ on a ``quadratic algebra'' containing
$\cscript$. We also show that a random variable $f$ can be ``quantized'' to form a self-adjoint operator $\fhat$ on $H$. We then employ $\fhat$ to define a quantum integral $\int fd\mutilde$. Various examples are given
\end{abstract}

\section{Introduction}  
This section presents an overview of the paper. Detailed definitions will be given in Sections~2, 3 and 4. The main arena for this study is a Hilbert space $H=L_2(\Omega ,\ascript ,\nu )$ where $(\Omega ,\ascript ,\nu )$ is a probability space. We think of $\Omega$ as the set of paths or trajectories or histories of a physical system. It is unusual to consider paths for a quantum system because such systems are not supposed to have well-defined trajectories, so paths are considered to be meaningless. However, paths are the basic ingredients of the histories approach to quantum mechanics \cite{djs10, dgt08, gt09, hal09, mocs05} and they appear in Feynman integrals and quantum gravity studies \cite{djs10, gt09, sor10}. Our attitude is that we are not abandoning the usual quantum formalism, but we are gleaning more information from this formalism by allowing the consideration of paths.

If $\rho$ is a density operator (state) on $H$ we define the \textit{decoherence functional}
$D_\rho\colon\ascript\times\ascript\to\complex$ by
\begin{equation*}
D_\rho (A,B)=\elbows{\rho\chi _B,\chi _A}
\end{equation*}
where $\chi _A$ is the characteristic function for $A\in\ascript$. We then define the \textit{quantum measure}
$\mu _\rho\colon\ascript\to\real ^+$ by $\mu _\rho (A)=D_\rho (A,A)$. If $f\colon\Omega\to\real$ is a random variable with $f\in H$ we define the ``quantization'' of $f$ to be a certain self-adjoint operator $\fhat$ on $H$. The
\textit{quantum integral} of $f$ is defined as
\begin{equation*}
\int fd\mu _\rho =\rmtr (\rho\fhat )
\end{equation*}
Properties of $D_\rho$, $\mu _\rho$ and $\int fd\mu _\rho$ are reviewed in Section~2.

A sequence of states $\rho _t$, $t=0,1,2,\ldots$, on $H$ that have certain properties is called a discrete quantum process and we call $\rho _t$ the local states for the process. If $\rho =\lim\rho _t$ exists, we call $\rho$ the global state for the process. For important cases, the global state does not exist and we must then work with the local states. In a natural way, the local states generate a sequence of quantum measures which in turn define a single quantum measure $\mu$ on the algebra of cylinder sets $\cscript$. It appears to be impossible to extend $\mu$ to a quantum measure on $\ascript$. An important problem is to extend $\mu$ to other physically relevant sets in a systematic way. We say that a set $A\in\ascript$ is \textit{suitable} if $\lim\elbows{\rho _t\chi _A,\chi _A}$ exists and is finite. We denote the collection of suitable sets by $\sscript$ and for $A\in\sscript$ we define
\begin{equation*}
\mutilde (A)=\lim\elbows{\rho _t\chi _A,\chi _A}
\end{equation*}
It is shown in Section~3 that $\sscript$ is a ``quadratic algebra'' that properly contains $\cscript$ and that $\mutilde$ is a quantum measure on $\sscript$ that extends $\mu$. It is also shown that the quantum integral extends in a natural way.

Section~4 considers finite unitary systems. Such a system is a set of unitary operators $U(s,r)$, $r\le s\in\positive$ on a position Hilbert space $\complex ^m$. The operator $U(s,r)$ describes the evolution of a finite-dimensional quantum system in discrete time-steps from time $r$ to time $s$. We call the elements of $S=\brac{0,1,\ldots ,m-1}$
\textit{sites} and we call infinite strings $\gamma =\gamma _0\gamma _1\cdots$, $\gamma _i\in S$
\textit{paths}. The \textit{path space} $\Omega$ is the set of all paths and the $n$-\textit{path space} $\Omega _n$ is the set of all $n$-paths $\gamma =\gamma _0\gamma _1\cdots\gamma _n$. The $n$-\textit{events} are sets in the power set $\ascript _n=2^{\Omega _n}$. Given an initial state $\psi\in\complex ^m$, the operators $U(s,r)$ define a decoherence functional $D_n\colon\ascript _n\times\ascript _n\to\complex$ in a natural way. The
\textit{decoherence matrix} is the $m^{n+1}\times m^{n+1}$ matrix with components
\begin{equation*}
D_n(\gamma ,\gamma ')=D_n\paren{\brac{\gamma},\brac{\gamma '}},\quad\gamma ,\gamma '\in\Omega _n
\end{equation*}
We can think of this matrix as an operator $\dhat _n$ on the $n$-path Hilbert space
$H_n=(\complex ^m)^{\otimes (n+1)}$. It is shown that $\dhat _n$ is a state on $H_n$ and the eigenvalues and eigenvectors of $\dhat _n$ are computed.

Section~5 shows how a finite unitary system can be employed to construct a discrete quantum process. Place the uniform probability distribution on $S$ and form the product measure on $\Omega =S\times S\times\cdots$ to obtain a probability space $(\Omega ,\ascript ,\nu )$. The path Hilbert space becomes $H=L_2(\Omega ,\ascript ,\nu )$. It is shown that the states $\dhat _t$, $t=0,1,2,\ldots$, generate a discrete quantum process $\rho _t$. We demonstrate that the event $A=$``the particle visits the origin'' as well as its complement $A'$ are elements of
$\sscript\smallsetminus\cscript$. An example of a two-site quantum random walk is explored and it is shown that the particle executes a periodic motion with period~4.

Section~6 considers quantum integrals. The operator $\fhat$ can be complicated and the expression $(\rho\fhat )$ can be difficult to evaluate. The eigenvalues and eigenvectors of $\fhat$ are found for a two-valued simple function
$f$. These are then employed to treat arbitrary simple functions. The quantum integral of an arbitrary random variable may then be computed by a limit process.

\section{Quantum Measures and Integrals} 
In a certain sense a quantum process is a generalization of a classical stochastic process. Moreover, quantum measures and integrals are generalizations of classical probability measures and classical expectations. For these reasons we begin with a short review of classical probability theory and then present a method of ``quantizing'' this structure.

A \textit{probability space} is a triple $(\Omega ,\ascript ,\nu )$ where $\Omega$ is a \textit{sample space} whose elements are \textit{sample points} or \textit{outcomes}, $\ascript$ is a $\sigma$-algebra of subsets of $\Omega$ whose elements are \textit{events} and $\nu$ is a measure on $\ascript$ satisfying $\nu (\Omega )=1$. For
$A\in\ascript$, $\nu (A)$ is interpreted as the probability that event $A$ occurs. We denote the set of measurable functions $f\colon\Omega\to\complex$ by $\mscript (\ascript )$. The first quantization step is to form the Hilbert space
\begin{equation*}
H=L_2(\Omega ,\ascript ,\nu )=\brac{f\in\mscript (\ascript )\colon\int\ab{f}^2d\nu <\infty}
\end{equation*}
with inner product $\elbows{f,g}=\int\fbar gd\nu$ and $\rmnorm\doubleab{f}=\elbows{f,f}^{1/2}$. We call real-valued functions $f\in H$ \textit{random variables}. If $f$ is a random variable, then by Schwarz's inequality we have
$\int\ab{f}d\nu\le\doubleab{f}$ so the expectation $E(f)=\int fd\nu$ exists and is finite. In general probability theory, random variables whose expectations are infinite or do not exist are considered, but for our purposes this more restricted concept is convenient.

The characteristic function $\chi _A$ of $A\in\ascript$ is a random variable with $\doubleab{\chi _A}=\nu (A)^{1/2}$ and we write $\chi _\Omega =1$. For $A,B\in\ascript$ we define the \textit{decoherence operator} $D(A,B)$ as the operator on $H$ defined by $D(A,B)=\ket{\chi _B}\bra{\chi _A}$. Thus, for $f\in H$ we have
\begin{equation*}
D(A,B)f=\elbows{\chi _A,f}\chi _B=\int _Afd\nu\chi _B
\end{equation*}
Of course, if $\nu (A)\nu (B)=0$ then $D(A,B)=0$. If $\nu (A)\nu (B)\ne 0$, it is easy to show that $D(A,B)$ is a rank~1 operator with $\doubleab{D(A,B)}=\nu (A)^{1/2}\nu (B)^{1/2}$. For $A\in\ascript$ we define the
$q$-\textit{measure operator} $\muhat (A)$ on $H$ by $\muhat (A)=D(A,A)$. Hence, for $f\in H$ we have
\begin{equation*}
\muhat (A)f=\elbows{\chi _A,f}\chi _A=\int _Afd\nu\chi _A
\end{equation*}
In particular, $\muhat (\Omega )f=E(f)1$. If $\nu (A)=0$, then $\mu (A)=0$ and if $\nu (A)\ne 0$, then $\muhat (A)$ is a positive (and hence, self-adjoint) rank~1 operator with $\doubleab{\muhat (A)}=\nu (A)$. Moreover, if $\nu (A)\ne 0$, then
\begin{equation*}
\frac{1}{\nu (A)}\muhat (A)=\frac{1}{\nu (A)}\ket{\chi _A}\bra{\chi _A}
\end{equation*}
is an orthogonal projection.

The map $D$ from $\ascript\times\ascript$ into the set of bounded operators $\bscript (H)$ on $H$ has some obvious properties:
\begin{list} {(\arabic{cond})}{\usecounter{cond}
\setlength{\rightmargin}{\leftmargin}}
\item If $A\cap B=\emptyset$, then $D(A\cup B,C)=D(A,C)+D(B,C)$ for all $C\in\ascript$ (additivity)
\item $D(A,B)^*=D(B,A)$ (conjugate symmetry)
\item $D(A,B)^2=\nu (A\cap B)D(A,B)$
\item $D(A,B)D(A,B)^*=\nu (A)\muhat (B)$, $D(A,B)^*D(A,B)=\nu (B)\muhat (A)$
\end{list}
Less obvious properties are given in the following theorem proved in \cite{gudapp}.

\begin{thm}       
\label{thm21}
{\rm (a)}\enspace $D\colon\ascript\times\ascript\to\bscript (H)$ is positive semidefinite in the sense that if
$A_i\in\ascript$, $c_i\in\complex$, $i=1,\ldots ,n$, then
\begin{equation*}
\sum _{i,j=1}^nD(A_i,A_j)c_i\cbar _j
\end{equation*}
is a positive operator.
{\rm (b)}\enspace If $A_1\subseteq A_2\subseteq\cdots$ is an increasing sequence in $\ascript$, then the continuity condition
\begin{equation*}
\lim D(A_i,B)=D(\cup A_i,B)
\end{equation*}
holds for every $B\in\ascript$ where the limit is in the operator norm topology.
\end{thm}

It follows from (1), (2) and Theorem~\ref{thm21}(b) that $A\mapsto D(A,B)$ and $B\mapsto D(A,B)$ are
operator-valued measures from $\ascript$ to $\bscript (H)$.

The map $\muhat\colon\ascript\to\bscript (H)$ need not be additive. For example, if $A,B\in\ascript$ are disjoint, then
\begin{align*}
\muhat (A\cup B)&=D(A\cup B,A\cup B)=\ket{\chi _{A\cup B}}\bra{\chi _{A\cup B}}
  =\ket{\chi _A+\chi _B}\bra{\chi _A+\chi _B}\\
  &=\ket{\chi _A}\bra{\chi _A}+\ket{\chi _B}\bra{\chi _B}+\ket{\chi _A}\bra{\chi _B}+\ket{\chi _B}\bra{\chi _A}\\
  &=\muhat (A)+\muhat (B)+2\rmre D(A,B)
\end{align*}
Thus, additivity is spoiled by the \textit{interference term} $2\rmre D(A,B)$. Because of this nonadditivity, we have that
$\muhat (A')\ne\muhat (\Omega )-\muhat (A)$ in general, where $A'$ is the complement of $A$. Moreover,
$A\subseteq B$ need not imply $\muhat (A)\le\muhat (B)$ in the usual order of self-adjoint operators. However,
$\muhat$ does satisfy the \textit{grade}-2 \textit{additivity} condition given in the next theorem which is proved in
\cite{gudapp}.

\begin{thm}       
\label{thm22}
{\rm (a)}\enspace $\muhat$ satisfies grade-2 additivity:
\begin{equation*}
\muhat (A\cup B\cup C)=\muhat (A\cup B)+\muhat (A\cup C)+\muhat (B\cup C)-\muhat (A)-\muhat (B)-\muhat (C)
\end{equation*}
whenever $A,B,C\in\ascript$ are mutually disjoint.
{\rm (b)}\enspace $\muhat$ satisfies the continuity conditions
\begin{align*}
\lim\muhat (A_i)&=\muhat (\cup A_i)\\
\lim\muhat (B_i)&=\muhat (\cap A_i)
\end{align*}
in the operator norm topology for any increasing sequence $A_i\in\ascript$ or decreasing sequence
$B_i\in\ascript$.
\end{thm}

If $\rho$ is a density operator (or state) on $H$ we define the \textit{decoherence functional}
$D_\rho\colon\ascript\times\ascript\to\complex$ by 
\begin{equation*}
D_\rho (A,B)=\rmtr\sqbrac{\rho D(A,B)}=\elbows{\rho\chi _B,\chi _A}
\end{equation*}
We interpret $D_\rho (A,B)$ as a measure of the interference between $A$ and $B$. The next result follows from Theorem~\ref{thm21}.

\begin{cor}       
\label{cor23}
{\rm (a)}\enspace $A\mapsto D_\rho (A,B)$ is a complex measure on $\ascript$.
{\rm (b)}\enspace If $A_1,\ldots ,A\in\ascript$, then the $n\times n$ matrix $D_\rho (A_i,A_j)$ is positive semidefinite.
\end{cor}

For the density operator $\rho$ on $H$ we define the $q$-measure $\mu _\rho\colon\ascript\to\real ^+$ by
\begin{equation*}
\mu _\rho (A)=\rmtr\sqbrac{\rho\muhat (A)}=\elbows{\rho\chi _A,\chi _A}
\end{equation*}
It can be shown that $\mu _\rho (\Omega )\le 1$ and that Theorem~\ref{thm22} holds with $\muhat$ replaced by
$\mu _\rho$ \cite{gudapp}. We interpret $\mu _\rho (A)$ as the $q$-probability or propensity of the event $A$ in the state $\rho$ \cite{gud10, sor94, sor07, sor11}.

We next introduce the second quantization step. Let $f$ be a nonnegative random variable. The
\textit{quantization} of $f$ is the operator $\fhat$ on $H$ defined by
\begin{equation*}
(\fhat g)(y)=\int\min\sqbrac{f(x),f(y)}g(x)d\nu (x)
\end{equation*}
It easily follows that $\doubleab{\fhat}\le\doubleab{f}$ so $\fhat$ is a bounded self-adjoint operator on $H$.

\begin{lem}       
\label{lem24}
If $f_1,f_2$ are nonnegative random variables with disjoint support, then $\fhat _1\fhat _2=\fhat _2\fhat _1=0$ and
\begin{equation*}
\doubleab{\fhat _1+\fhat _2}=\max\sqbrac{\doubleab{\fhat _1},\doubleab{\fhat _2}}
\end{equation*}
\end{lem}
\begin{proof}
For every $g\in H$ we have by Fubini's theorem that
\begin{align*}
(\fhat _1\fhat _2g)(y)&=\fhat _1\sqbrac{\int\min\sqbrac{f_2(x),f_2(z)}g(x)d\nu (x)}(y)\\
  &=\int\min\sqbrac{f_1(z),f_1(y)}\brac{\int\min\sqbrac{f_2(x),f_2(z)}g(x)d\nu (x)}d\nu (z)\\
  &=\int\brac{\int\min\sqbrac{f_1(z),f_1(y)}\min\sqbrac{f_2(x),f_2(z)}d\nu (z)}g(x)d\nu (x)\\
  &=0
\end{align*}
where the last equality follows from $f_1(z)f_2(z)=0$. The second statement now follows.
\end{proof}

If $f$ is an arbitrary random variable, we have that $f=f^+-f^-$ where $f^+(x)=\max\sqbrac{f(x),0}$ and
$f^-(x)=-\min\sqbrac{f(x),0}$. Then $f^+,f^-\ge 0$, $f^+f^-=0$ and we define the bounded self-adjoint operator $\fhat$ by $\fhat=f^{+\wedge}-f^{-\wedge}$. It follows from Lemma~\ref{lem24} that
$\doubleab{\fhat}=\max\sqbrac{\doubleab{f^+},\doubleab{f^-}}$. The next result summarizes some of the important properties of $\fhat$ \cite{gudapp}.

\begin{thm}       
\label{thm25}
{\rm (a)}\enspace For any $A\in\ascript$, $\chihat _A=\ket{\chi _A}\bra{\chi _A}=\muhat (A)$.
{\rm (b)}\enspace For any $\alpha\in\real$, $(\alpha f)^\wedge =\alpha\fhat$.
{\rm (c)}\enspace If $f\ge 0$, then $\fhat$ is a positive operator.
{\rm (d)}\enspace If $0\le f_1\le f_2\le\cdots$ is an increasing sequence of random variables converging in norm to a random variable $f$, then $\lim\fhat _i=\fhat$ in the operator norm topology.
{\rm (e)}\enspace If $f,g,h$ are random variables with disjoint supports, then
\begin{equation*}
(f+g+h)^\wedge =(f+g)^\wedge+(f+h)^\wedge +(g+h)^\wedge -\fhat -\ghat -\hhat
\end{equation*}
\end{thm}

Let $\rho$ be a density operator on $H$ and let $\mu _\rho (A)=\rmtr\paren{\rho\muhat (A)}$ be the corresponding
$q$-measure. If $f$ is a random variable we define the $q$-\textit{integral} (or $q$-\textit{expectation}) of $f$ with respect to $\mu _\rho$ as
\begin{equation*}
\int fd\mu _\rho =\rmtr (\rho\fhat )
\end{equation*}
As usual, for $A\in\ascript$ we define
\begin{equation*}
\int _Afd\mu _\rho =\int\chi _Afd\mu _\rho
\end{equation*}
The next result follows from Theorem~\ref{thm25}.

\begin{cor}       
\label{cor26}
{\rm (a)}\enspace For every $A\in\ascript$, $\int\chi _Ad\mu _\rho=\mu _\rho (A)$.
{\rm (b)}\enspace For every $\alpha\in\real$, $\int\alpha fd\mu _\rho=\alpha\int fd\mu _\rho$.
{\rm (c)}\enspace If $f\ge 0$, then $\int fd\mu _\rho\ge 0$.
{\rm (d)}\enspace If $f_i\ge 0$ is an increasing sequence of random variables converging in norm to a random variable $f$, then $\lim\int f_id\mu _\rho =\int fd\mu _\rho$.
{\rm (e)}\enspace If $f,g,h$ are random variables with disjoint supports, then
\begin{align*}
\int (f+g+h)d\mu _\rho&=\int (f+g)d\mu _\rho +\int (f+h)d\mu _\rho +\int (g+h)d\mu _\rho\\
  &\quad -\int fd\mu _\rho -\int gd\mu _\rho -\int hd\mu _\rho
\end{align*}
\end{cor}

The next result is called the \textit{tail-sum formula} and gives a justification for calling $\int fd\mu _\rho$ a
$q$-\textit{integral} \cite{gud09, gudapp}. The classical tail-sum formula is quite useful in traditional probability theory
\cite{gud09}.

\begin{thm}       
\label{thm27}
If $f\ge 0$ is a random variable, then
\begin{equation*}
\int fd\mu _\rho =\int _0^\infty\mu _\rho\brac{x\colon f(x)>\lambda}d\lambda
\end{equation*}
where $d\lambda$ denotes Lebesgue measure on $\real$.
\end{thm}

It follows from Theorem~\ref{thm27} that if $f$ is an arbitrary random variable, then
\begin{equation*}
\int fd\mu _\rho =\int _0^\infty\mu _\rho\brac{x\colon f(x)>\lambda}d\lambda
  -\int _0^\infty\mu _\rho\brac{x\colon f(x)<-\lambda}d\lambda
\end{equation*}

\section{Discrete Quantum Processes} 
Let $(\Omega ,\ascript ,\nu )$ be a probability space and let $\ascript _t\subseteq\ascript$, $t=0,1,2\ldots$, be an increasing sequence of $\sigma$-algebras such that $\ascript$ is the smallest $\sigma$-algebra containing
$\cup\ascript _t$. We then say that $\ascript _t=0,1,2\ldots$, \textit{generates} $\ascript$. Let $\nu _t$ be the restriction of $\nu$ to $\ascript _t$. We think of the probability space $(\Omega ,\ascript _t,\nu _t)$as a classical description of a physical system until a discrete time $t$. A corresponding quantum description takes place in the closed subspace $H_t=L_2(\Omega ,\ascript _t,\nu _t)$ of $H=L_2(\Omega ,\ascript ,\nu )$ $t=0,1,2\ldots\,$. A sequence of density operators $\rho _t$ on $H_t$, $t=0,1,2,\ldots$, is \textit{consistent} if $D_{\rho _{t+1}}(A,B)=D_{\rho _t}(A,B)$ for every $A,B\in\ascript _t$. In particular, we then have that $\mu _{\rho _{t+1}}(A)=\mu _{\rho _t}(A)$ for all $A\in\ascript _t$,
$t=0,1,2,\ldots\,$. We call a consistent sequence $\rho _t$, $t=0,1,2,\ldots$, a \textit{discrete} $q$-\textit{process} and we call $\rho _t$ the \textit{local states} for the process. If each $\ascript _t$ has finite cardinality, we call a consistent sequence $\rho _t$ a \textit{finite} $q$-\textit{process}. Notice for a finite $q$-process that each of the subspaces
$H_t$ is finite-dimensional.

Let $\rho _t$, $t=0,1,2,\ldots$, be a discrete $q$-process. We then have an increasing sequence of closed subspaces
$H_0\subseteq H_1\subseteq H_2\subseteq\cdots\subseteq H$. We now extend $\rho _t$ from $H_t$ to a state on
$H$ by defining $\rho _tf=0$ for all $f\in H_t^\perp$. If $\lim\rho _t$ exists in the strong operator topology, then the limit $\rho$ is a state on $H$. In this case we call $\rho$ the \textit{global state} for the $q$-process $\rho _t$. If the global state
$\rho$ exists, then for every $A,B\in\ascript _t$ we have
\begin{equation}         
\label{eq31}
D_\rho (A,B)=\lim _{n\to\infty}D_{\rho _n}(A,B)=D_{\rho _t}(A,B)
\end{equation}
We call $D_\rho$ and $\mu _\rho$ the \textit{global decoherence functional} and \textit{global} $q$-\textit{measure} for the $q$-process, respectively. Equation~\eqref{eq31} shows that $D_\rho$ and $\mu _\rho$ extend all the
$D_{\rho _t}$ and $\mu _{\rho _t}$ from $\ascript _t$ to $\ascript$, $t=0,1,2,\ldots$, respectively. It follows from the work in \cite{djsu10, gsapp} that $\lim\rho _t$ may not exist in which case we would not have a global state and hence no global decoherence functional or global $q$-measure. Moreover, we would not have a global integral
$\int fd\mu _\rho$. In this case we are forced to work with the local states $\rho _t$, $t=0,1,2,\ldots$, and this is what we now explore.

A collection $\qscript$ of subsets of a set $X$ is a \textit{quadratic algebra} if $\emptyset ,X\in\qscript$ and if
$A,B,C\in\qscript$ are mutually disjoint and $A\cup B, A\cup C,B\cup C\in\qscript$, then $A\cup B\cup C\in\qscript$. Of course, an algebra of subsets of $X$ is a quadratic algebra. However, there are examples of quadratic algebras that are not closed under complementation, union or intersection \cite{gsapp}. A general $q$-\textit{measure} is a nonnegative grade-2 set function $\mu$ on a quadratic algebra $\qscript$. That is, $\mu\colon\qscript\to\real ^+$ and if $A,B,C\in\qscript$ are mutually disjoint with $A\cup B,A\cup C,B\cup C\in\qscript$, then
\begin{equation*}
\mu (A\cup B\cup C)=\mu (A\cup B)+\mu (A\cup C)+\mu (B\cup C)-\mu (A)-\mu (B)-\mu (C)
\end{equation*}

Let $\rho _t$, $t=0,1,2,\ldots$, be a discrete $q$-process. Defining $\cscript (\Omega )=\cup\ascript _t$, it is clear that $\cscript (\Omega )$ is an algebra of subsets of $\Omega$. For $A\in\cscript (\Omega )$ we have that $A\in\ascript _t$ for some $t\in\positive$ and we define $\mu (A)=\mu _{\rho _t}(A)$. To show that $\mu$ is well-defined, suppose that $A\in\ascript _t\cap\ascript _{t'}$. We can assume without loss of generality that $\ascript _t\subseteq\ascript _{t'}$. Hence, $\mu _{\rho _{t'}}(A)=\mu _{\rho _t}(A)$ so $\mu$ is well-defined. It easily follows that
$\mu\colon\cscript (\Omega )\to\real ^+$ is a $q$-measure. It appears to be impossible to extend $\mu$ to a
$q$-measure on $\ascript$ in general. In fact, it is shown in \cite{djsu10, gsapp} that in general $\mu$ cannot be extended to a continuous $q$-measure on $\ascript$. However, we can extend $\mu$ to a $q$-measure on a larger quadratic algebra than $\cscript (\Omega )$ and this quadratic algebra contains physically relevant sets that are not in $\cscript (\Omega )$. A set $A\in\ascript$ is \textit{suitable} if $\lim\rmtr (\rho _t\chihat _A)$ exists and is finite. We denote the collection of suitable sets by $\sscript (\Omega )$ and for $A\in\sscript (\Omega )$ we define
$\mutilde (A)=\lim\rmtr (\rho _t\chihat _A)$.

\begin{thm}       
\label{thm31}
$\sscript (\Omega )$ is a quadratic algebra that contains $\cscript (\Omega )$ and $\mutilde$ is a $q$-measure on
$\sscript (\Omega )$ that extends $\mu$.
\end{thm}
\begin{proof}
If $A\in\cscript (\Omega )$ then $A\in\ascript _t$ for some $t\in\positive$. Since $\mu _{\rho _{t'}}(A)=\mu _{\rho _t}(A)$ for all $t'\ge t$ we have
\begin{equation*}
\lim\rmtr (\rho _{t'}\chihat _A)=\rmtr (\rho _t\chihat _A)=\mu _{\rho _t}(A)=\mu (A)
\end{equation*}
Hence, $A\in\sscript (\Omega )$ and $\mutilde (A)=\mu (A)$. We conclude that
$\cscript (\Omega )\subseteq\sscript (\Omega )$ and $\mutilde$ extends $\mu$ to $\sscript (\Omega )$. To show that $\sscript (\Omega )$ is a quadratic algebra, suppose $A,B,C\in\sscript (\Omega )$ are mutually disjoint with
$A\cup B,A\cup C,B\cup C\in\sscript (\Omega )$. Applying Theorem~\ref{thm25}(e) we obtain
\begin{align*}
\lim\rmtr (\rho _t\chihat _{A\cup B\cup C})&=\lim\rmtr\sqbrac{\rho _t(\chi _A+\chi _B+\chi _C)^\wedge}\\
  &=\lim\rmtr (\rho _t\chihat _{A\cup B})+\lim\rmtr (\rho _t\chihat_{A\cup C})+\lim\rmtr (\rho _t\chihat _{B\cup C})\\
  &\quad -\lim\rmtr (\rho _t\chihat _A)-\lim\rmtr (\rho _t\chihat _B)-\lim\rmtr (\rho _t\chihat _A)
\end{align*}
We conclude that $A\cup B\cup C\in\sscript (\Omega )$ and that $\mutilde$ is a $q$-measure on $\sscript (\Omega )$.
\end{proof}

A random variable $f\in H$ is \textit{integrable} for $\rho _t$ if $\lim\rmtr (\rho _t\fhat )$ exists and is finite. If $f$ is integrable we define
\begin{equation*}
\int fd\mutilde =\lim\rmtr (\rho _t\fhat )
\end{equation*}
Notice that if $A\in\sscript (\Omega )$ then $\chi _A$ is integrable and
\begin{equation*}
\int\chi _Ad\mutilde =\mutilde (A)
\end{equation*}
The proof of the next result is similar to the proof of Theorem~\ref{thm31} and follows from Corollary~\ref{cor26}.

\begin{thm}       
\label{thm32}
{\rm (a)}\enspace If $f$ is integrable and $\alpha\in\real$, then $\alpha f$ is integrable and
$\int\alpha fd\mutilde =\alpha\int fd\mutilde$
{\rm (b)}\enspace If $f$ is integrable with $f\ge 0$, then $\int fd\mutilde\ge 0$.
{\rm (c)}\enspace If $f,g,h$ are integrable with disjoint support and $f+g$, $f+h$, $g+h$ are integrable, then $f+g+h$ is integrable and
\begin{align*}
\int (f+g+h)d\mutilde&=\int (f+g)d\mutilde +\int (f+h)d\mutilde +\int (g+h)d\mutilde\\
  &\quad -\int fd\mutilde -\int gd\mutilde -\int hd\mutilde
\end{align*}
\end{thm}

\section{Finite Unitary Systems} 
Finite unitary systems and their relationship to finite $q$-processes have been studied in the past
\cite{djsu10, gsapp}. They are closely related to the histories approach to quantum mechanics
\cite{hal09,mocs05,sor07,sor10}. In this section we study finite unitary systems and in Section~5 we employ them to construct finite $q$-processes.

Let $\complex ^m$ be the $m$-dimensional Hilbert space with elements
\begin{equation*}
f\colon\brac{0,1,\ldots ,m-1}\to\complex
\end{equation*}
and inner product
\begin{equation*}
\elbows{f,g}=\sum _{j=0}^{m-1}\overline{f(j)}g(j)
\end{equation*}
We call $\complex ^m$ the \textit{position HIlbert space} and denote the standard basis on $\complex ^m$ by
$e_0,e_1,\ldots ,e_{m-1}$. A \textit{finite unitary system} is a collection of unitary operators $U(s,r)$,
$r\le s\in\positive$ on $\complex ^m$ such that $U(r,r)=I$ and
\begin{equation*}
U(t,r)=U(t,s)U(s,r)
\end{equation*}
for $r\le s\le t\in\positive$. If $U(s,r)$, $r\le s\in\positive$, is a finite unitary system, then we have the unitary operators $U(n+1,n),n\in\positive$ such that
\begin{equation}         
\label{eq41}
U(s,r)=U(s,s-1)U(s-1,s-2)\cdots U(r+1,r)
\end{equation}
Conversely, if $U(n+1,n),n\in\positive$, are unitary operators on $\complex ^m$, then defining $U(r,r)=I$ and for 
$r<s$ defining $U(s,r)$ by \eqref{eq41} we have that $U(s,r)$, $r\le s\in\positive$, is a finite unitary system. A finite unitary system $U(s,r)$ $r\le s\in\positive$ is \textit{stationary} if there is a unitary operator $U$ on $\complex ^m$ such that $U(s,r)=U^{s-r}$ for all $r\le s\in\positive$. In this case $U(n+1,n)=U$ for all $n\in\positive$.

We call the elements of $S=\brac{0,1,\ldots ,m-1}$ \textit{sites} or \textit{positions} and we call infinite strings
$\gamma =\gamma _0\gamma _1\gamma _2\cdots$, $\gamma _i\in S$, \textit{paths} or \textit{trajectories} or
\textit{histories}. The \textit{path} or \textit{sample space} is
\begin{equation*}
\Omega =\brac{\gamma\colon\gamma\hbox{ a path}}
\end{equation*}
We also call finite strings $\gamma _0\gamma _1\cdots\gamma _n$ $n$-\textit{paths} and
\begin{equation*}
\Omega _n=\brac{\gamma\colon\gamma\hbox{ an }n\hbox{-path}}
\end{equation*}
is the $n$-\textit{path space} on $n$-\textit{sample space}. Notice that the cardinality $\ab{\Omega _n}=m^{n+1}$. We call the elements of $\ascript _n=2^{\Omega _n}$ $n$-\textit{events}.

A finite unitary system $U(s,r)$ describes the evolution of a finite-dimen\-sional quantum system and the projections
$P(i)=\ket{e_i}\bra{e_i}$, $i=0,1,\ldots ,m-1$, describe the position. The $n$-path $\gamma\in\Omega _n$ is described by the operator $C_n(\gamma )$ on $\complex ^m$ given by
\begin{equation}         
\label{eq42}
C_n(\gamma )=P(\gamma _n)U(n,n-1)P(\gamma _{n-1})U(n-1,n-2)\cdots P(\gamma _1)U(1,0)P(\gamma _0)
\end{equation}
Defining $b(\gamma )$ by
\begin{align}         
\label{eq43}
b(\gamma )&=\elbows{e_{\gamma _n},U(n,n-1)e_{\gamma _{n-1}}}
  \elbows{e_{\gamma _{n-1}},U(n-1,n-2)e_{\gamma _{n-2}}}\notag\\
  &\quad\cdots\elbows{e_{\gamma _1},U(1,0)e_{\gamma _0}}
\end{align}
Equation~\eqref{eq42} becomes
\begin{equation}         
\label{eq44}
C_n(\gamma )=b(\gamma )\ket{e_{\gamma _n}}\bra{e_{\gamma _0}}
\end{equation}

\begin{lem}       
\label{lem41}
For $i=0,1,\ldots ,m-1$ we have
\begin{equation*}
\sum _{\gamma\in\Omega _n}\brac{\ab{b(\gamma )}^2\colon\gamma _0=i}
   =\sum _{\gamma\in\Omega _n}\brac{\ab{b(\gamma )}^2\colon\gamma _n=i}=1
\end{equation*}
\end{lem}
\begin{proof}
The result follows from
\begin{align*}
\ab{b(\gamma )}^2&=\ab{\elbows{e_{\gamma _n},U(n,n-1)e_{\gamma _{n-1}}}}^2
  \ab{\elbows{e_{\gamma _{n-1}},U(n-1,n-2)e_{\gamma _{n-2}}}}^2\\
  &\quad\cdots\ab{\elbows{e_{\gamma _1},U(1,0)e_{\gamma _0}}}^2
\end{align*}
and calculating the designated sums.
\end{proof}
 If $\psi\in\complex ^m$, $\doubleab{\psi}=1$, we define the \textit{amplitude} of $\gamma\in\Omega _n$ by
 $a _\psi (\gamma )=b(\gamma )\psi (\gamma _0)$. We interpret $\ab{a_\psi (\gamma )}^2$ as the probability of the path $\gamma$ with initial distribution $\psi$. The next result shows that these probabilities sum to 1.

\begin{cor}       
\label{cor42}
For the path space $\Omega _n$ we have
\begin{equation*}
\sum _{\gamma\in\Omega _n}\ab{a_\psi (\gamma )}^2=1
\end{equation*}
\end{cor}
\begin{proof}
By Lemma~\ref{lem41} we have
\begin{align*}
\sum _{\gamma\in\Omega _n}\ab{a_\psi (\gamma )}^2
  &=\sum _{\gamma _0}\sum _{\gamma _n,\ldots ,\gamma _1}\ab{a_\psi (\gamma )}^2
  =\sum _{\gamma _0}\ab{\psi (\gamma _0)}^2\sum _{\gamma _n,\ldots ,\gamma _1}\ab{b(\gamma )}^2\\
  &=\sum _{\gamma _0}\ab{\psi (\gamma _0)}^2=1\qedhere
\end{align*}
\end{proof}

It is interesting to note that
\begin{align*}
a_\psi&(\gamma )\\
  &=\elbows{e_{\gamma _n}\otimes\cdots\otimes e_{\gamma _0},U(n,n-1)\otimes\cdots\otimes
  U(1,0)\otimes Ie_{\gamma _{n-1}}\otimes\cdots\otimes e_{\gamma _0}\otimes\psi}
\end{align*}
and when the system is stationary with evolution operator $U$ we have
\begin{equation*}
a_\psi (\gamma )=\elbows{e_{\gamma _n}\otimes\cdots\otimes e_{\gamma _0},U^{\otimes n}
  \otimes Ie_{\gamma _{n-1}}\otimes\cdots\otimes e_{\gamma _0}\otimes\psi}
\end{equation*}

The operator $C_n(\gamma ')^*C_n(\gamma )$ describes the interference between the two paths
$\gamma ,\gamma '\in\Omega _n$. More precisely, $C_n(\gamma ')^*C_n(\gamma )$ describes the interference between two particles, one moving along path $\gamma$ and the other along path $\gamma '$. Applying
\eqref{eq44} we see that
\begin{equation}         
\label{eq45}
C_n(\gamma ')^*C_n(\gamma )=\overline{b(\gamma ')}b(\gamma )\ket{e _{\gamma '_0}}\bra{e_{\gamma '_0}}
  \delta _{\gamma _n,\gamma '_n}
\end{equation}
For $A\in\ascript _n$ the \textit{class operator} $C_n(A)$ is
\begin{equation*}
C_n(A)=\sum _{\gamma\in A}C_n(\gamma )
\end{equation*}
It is clear that $A\mapsto C_n(A)$ is an operator-valued measure on $\ascript _n$ satisfying
$C_n(\Omega _n)=U(n,0)$. Indeed, by \eqref{eq43} and \eqref{eq44} we have
\begin{equation*}
C_n(\Omega _n)=\sum _{\gamma\in\Omega _n}C_n(\gamma )=\sum _{\gamma\in\Omega _n}
  \elbows{e_{\gamma _n},U(n,0)e_{\gamma _0}}\ket{e_{\gamma _n}}\bra{e_{\gamma _0}}=U(n,0)
\end{equation*}
The \textit{decoherence functional} $D_n\colon\ascript _n\times\ascript _n\to\complex$ is defined by
\begin{equation*}
D_n(A,B)=\elbows{C_n(A)^*C_n(B)\psi ,\psi}
\end{equation*}
where $\psi\in\complex ^m$, $\doubleab{\psi}=1$, is the initial state. It is clear that $A\mapsto D_n(A,B)$ is a
complex-valued measure on $\ascript _n$ with $D_n(\Omega _n,\Omega _n)=1$. It is well known that for any
$A_1,\ldots ,A_k\in\ascript _n$, $D_n(A_i,A_j)$ is a positive semidefinite matrix \cite{gt09, sor94, sor07}.

The $q$-measure $\mu _n\colon\ascript _n\to\real ^+$ is defined by $\mu _n(A)=D_n(A,A)$. It is well known that
$\mu _n$ indeed satisfies the grade-2 additivity condition required for a $q$-measure \cite{gt09, sor94, sor07}.

The $n$-\textit{distribution} given by
\begin{equation*}
p_n(i)=\mu _n\paren{\brac{\gamma\in\Omega _n\colon\gamma _n=i}}
\end{equation*}
is interpreted as the probability that the system is at site $i$ at time $n$. The next result shows that $p_n(i)$ gives the usual quantum distribution.

\begin{thm}       
\label{thm43}
For $i=0,1,\ldots ,m$ we have
\begin{equation*}
p_n(i)=\ab{\sum _{\gamma _n=i}a_\psi (\gamma )}^2=\ab{\elbows{e_i,U(n,0)\psi}}^2
\end{equation*}
\end{thm}
\begin{proof}
Letting $A=\brac{\gamma\in\Omega _n\colon\gamma _n=i}$ we have by \eqref{eq45} that
\begin{align*}
p_n(i)&=D_n(A,A)=\elbows{C_n(A)^*C_n(A)\psi ,\psi}\\
  &=\sum\brac{\elbows{C(\gamma ')^*C(\gamma )\psi ,\psi}\colon\gamma '_n=\gamma _n=i}\\
  &=\sum\brac{\overline{b(\gamma )}b(\gamma ')\overline{\psi (\gamma _0)}\psi (\gamma '_0)\colon\gamma '_n
  =\gamma _n=i}\\
  &=\ab{\sum _{\gamma _n=i}b(\gamma )\psi (\gamma _0)}^2=\ab{\sum _{\gamma _n=i}a_\psi (\gamma )}^2
\end{align*}
Applying \eqref{eq43} gives
\begin{align*}
\sum _{\gamma _n=i}b(\gamma )\psi (\gamma _0)
  &=\sum _{\gamma _0}\elbows{e_i,U(n,0)e_{\gamma _0}}\psi (\gamma _0)\\
  &=\sum _{\gamma _0}\elbows{U(n,0)^*e_i,e_{\gamma _0}}\elbows{e_{\gamma _0},\psi}\\
  &=\elbows{e_i,U(n,0)\psi}
\end{align*}
The result now follows.
\end{proof}

Corresponding to an initial state $\psi\in\complex ^m$, $\doubleab{\psi}=1$, the $n$-\textit{decoherence matrix} is
$D_n(\gamma ,\gamma ')=D_n\paren{\brac{\gamma},\brac{\gamma '}}$. We have by \eqref{eq45} that
\begin{align}         
\label{eq46}
D_n(\gamma ,\gamma ')&=\elbows{C_n(\gamma ')^*C_n(\gamma )\psi ,\psi}
  =b(\gamma )\overline{b(\gamma ')}\psi (\gamma _0)\overline{\psi (\gamma '_0)}\delta _{\gamma _n,\gamma '_n}
  \notag\\
  &=a_\psi (\gamma )\overline{a_\psi (\gamma ')}\delta _{\gamma _n,\gamma '_n}
\end{align}
Notice that
\begin{equation*}
\mu _n(\gamma )=D_n(\gamma ,\gamma )=\ab{a_\psi (\gamma )}^2
\end{equation*}
and by Corollary~\ref{cor42},
\begin{equation*}
\sum _{\gamma\in\Omega _n}\mu _n(\gamma )=1
\end{equation*}
We also have
\begin{equation*}
\mu _n\paren{\brac{\gamma ,\gamma '}}=\mu _n(\gamma )+\mu _n(\gamma ')+2\rmre D_n(\gamma ,\gamma ')
\end{equation*}
Finally, notice that
\begin{align*}
D_n(A,B)&=\sum\brac{D_n(\gamma ,\gamma ')\colon\gamma\in A,\gamma '\in B}\\
  &=\sum \brac{a_\psi (\gamma )\overline{a_\psi (\gamma ')}\delta _{\gamma _n,\gamma '_n}
  \colon\gamma\in A,\gamma '\in B}
\end{align*}
and hence,
\begin{equation}         
\label{eq47}
\mu _n(A)=D_n(A,A)=\sum _{\gamma ,\gamma '\in A}D_n(\gamma ,\gamma ')
  =\sum _{\gamma ,\gamma '\in A}a_\psi (\gamma )
  \overline{a_\psi (\gamma ')}\delta _{\gamma _n,\gamma '_n}
\end{equation}
 Define the $n$-\textit{path Hilbert space} $H_n=(\complex ^m)^{\otimes (n+1)}$. We associate
 $\gamma\in\Omega _n$ with the unit vector in $H_n$ given by
 \begin{equation*}
e_{\gamma _n}\otimes e_{\gamma _{n-1}}\otimes\cdots\otimes e_{\gamma _0}
\end{equation*}
We can think of $H_n$ as the set $\brac{\phi\colon\Omega _n\to\complex}$ with the usual inner product. Then
$\gamma\in\Omega _n$ corresponds to $\chi _{\brac{\gamma}}$ and the matrix with components
$D_n(\gamma ,\gamma ')$ corresponds to the operator
\begin{equation*}
\paren{\dhat _n\phi}(\gamma )=\sum _{\gamma '}D_n(\gamma ,\gamma ')\phi (\gamma ')
\end{equation*}

\begin{thm}       
\label{thm44}
The operator $\dhat _n$ is a state on $H_n$.
\end{thm}
\begin{proof}
It follows from \eqref{eq46} that $\dhat _n$ is a positive operator \cite{gsapp}. By Corollary~\ref{cor42} we have
\begin{equation*}
\rmtr\paren{\dhat _n}=\sum _{\gamma\in\Omega _n}D_n(\gamma ,\gamma )
  =\sum _{\gamma\in\Omega _n}\ab{a_\psi (\gamma )}^2=1
\end{equation*}
Hence, $\dhat _n$ is a trace 1 positive operator so $\dhat _n$ is a state on $H_n$.
\end{proof}

For every $A\in\ascript _n$ we have the vector $\ket{\chi _A}=\sum\brac{\gamma\colon\gamma\in A}$.

\begin{lem}       
\label{lem45}
The decoherence functional satisfies
\begin{equation*}
D_n(A,B)=\rmtr\paren{\ket{\chi _B}\bra{\chi _A}\dhat _n}
\end{equation*}
\end{lem}
\begin{proof}
For $A,B\in\ascript _n$ we have
\begin{align*}
\rmtr\paren{\ket{\chi _B}\bra{\chi _A}\dhat _n}
  &=\sum _{\gamma\in\Omega _n}\elbows{\ket{\chi _B}\bra{\chi _A}\dhat _n\gamma ,\gamma}\\
  &=\sum _{\gamma\in\Omega _n}\elbows{\dhat _n\gamma ,\ket{\chi _A}\bra{\chi _B}\gamma}
  =\sum _{\gamma\in B}\elbows{\dhat _n\gamma ,\chi _A}\\
  &=\sum\brac{\elbows{\dhat _n\gamma ,\gamma '}\colon\gamma\in B,\gamma '\in A}\\
  &=\sum\brac{D_n(\gamma ',\gamma )\colon\gamma '\in A,\gamma\in B}=D_n(A,B)\qedhere
\end{align*}
\end{proof}

The next result characterizes the eigenvalues and eigenvectors of the operator $\dhat _n$. Let
$\gamma ^0,\gamma ^1,\ldots ,\gamma ^{m-1}\in\Omega _n$ be the $n$-paths given by$\gamma ^i=00\cdots 0i$, $i=0,1,\ldots ,m-1$.

\begin{thm}       
\label{thm46}
The nonzero eigenvalues of $\dhat _n$ have the form
\begin{equation*}
\lambda _i=\sum\brac{\mu _n(\gamma )\colon\gamma _n=i}
\end{equation*}
for $i=0,1,\ldots ,m-1$ and corresponding eigenvectors $u^i\in H_n$ have entries
\begin{equation*}
u^i(\gamma )=\overline{a_\psi (\gamma ^i)}a_\psi (\gamma )\delta _{i,\gamma _n}
\end{equation*}
\end{thm}
\begin{proof}
For $i=0,1,\ldots ,m-1$ we have
\begin{align*}
(\dhat _nu^i)(\gamma )&=\sum _{\gamma '\in\Omega _n}D_n(\gamma ,\gamma ')u^i(\gamma ')\\
  &=\sum _{\gamma '\in\Omega _n}a_\psi (\gamma )\overline{a_\psi (\gamma ')}\delta _{\gamma _n,\gamma '_n}
  \overline{a_\psi (\gamma ^i)}a_\psi (\gamma )\delta _{i,\gamma '_n}\\
  &=\overline{a_\psi (\gamma ^i)}a_\psi (\gamma )\delta _{i,\gamma _n}
  \sum _{\gamma '\in\Omega _n}\ab{a_\psi (\gamma '}^2\delta _{i,\gamma '_n}\\
  &=\sqbrac{\sum\brac{\ab{a_\psi (\gamma )}^2\colon\gamma _n=i}}u^i(\gamma )\\
  &=\sqbrac{\sum\brac{\mu _n(\gamma )\colon\gamma _n=i}}u^i(\gamma )
\end{align*}
This shows that $u^i$, $i=0,1,\ldots ,m-1$ are eigenvectors of $D_n$ with corresponding eigenvalues $\lambda _i$. By Corollary~\ref{cor42} and the fact that $\mu _n(\Omega _n)=1$ we have that
$\sum\lambda _i=1=\rmtr\paren{\dhat _n}$. Hence, $\lambda _0,\ldots ,\lambda _{m-1}$ include all nonzero eigenvalues of $\dhat _n$.
\end{proof}
Assuming that $u^i\ne 0$ in Theorem~\ref{thm46}, $i=0,\ldots ,m-1$, let $v_i=u^i/\doubleab{u^i}$ be the corresponding unit eigenvectors. We then have the spectral resolution $\dhat _n=\sum\lambda _iP_{v_i}$ where
$P_{v_i}$ is the projection onto the subspace spanned by $v_i$, $i=0,1,\ldots ,m-1$.

\begin{cor}       
\label{cor47}
For the eigenvalues $\lambda _i$ of $\dhat _n$ we have for every $A\in\ascript _n$ that
\begin{equation*}
\mu _n(A)=\sum _{i=0}^{m-1}\lambda _i\ab{\sum _{\gamma\in A}\elbows{\chi _{\brac{\gamma}},v_i}}^2
\end{equation*}
\end{cor}
\begin{proof}
By Lemma~\ref{lem45} and Theorem~\ref{thm46} we have
\begin{align*}
\mu _n(A)&=\rmtr\paren{\ket{\chi _A}\bra{\chi _A}\dhat _n}
  =\sum _{i=0}^{m-1}\lambda _i\rmtr\paren{\ket{\chi _A}\bra{\chi _A}P_{v_i}}\\
  &=\sum _{i=0}^{m-1}\lambda _i\ab{\elbows{\chi _A,v_i}}^2
  =\sum _{i=0}^{m-1}\lambda _i\ab{\sum _{\gamma\in A}\elbows{\chi _{\brac{\gamma}},v_i}}^2\qedhere
\end{align*}
\end{proof}

\section{Finite Unitary Processes} 
This section shows how a finite unitary system can be employed to construct a finite unitary $q$-process. As in Section~4, let $S=\brac{0,1,\ldots ,m-1}$ be a set of sites and let $\Omega =S\times S\times\cdots$ be the set of all paths. Place the discrete topology on $S$ and endow $\Omega$ with the product topology. Then $\Omega$ is a compact Hausdorff space. Let $\ascript$ be the $\sigma$-algebra generated by the open sets in $\Omega$ and let
$\cscript$ be the algebra of cylinder sets
\begin{equation*}
A_0\times A_1\times\cdots\times A_n\times S\times S\times\cdots
\end{equation*}
$A_i\subseteq S$, $i=0,1,\ldots n$. Let $p_0$ be the probability measure on $S$ given by $p_0(A)=\ab{A}/m$,
$A\subseteq S$ and $p_1$ be the probability measure on $\cscript$ given by
\begin{equation*}
p_1(A_0\times A_1\times\cdots\times A_n\times S\times S\times\cdots )=p_0(A_0)p_0(A_1)\cdots p_0(A_n)
\end{equation*}
Then $p_1$ is countably additive on $\cscript$ so by the Kolmogorov extension theorem $p_1$ has a unique extension to a probability measure $\nu$ on $\ascript$. We define the \textit{path Hilbert space} by
$H=L_2(\Omega ,\ascript ,\nu )$.\medskip

\noindent\textbf{Example 1.}\enspace Let $A\in\ascript$ be the event ``the system visits the origin.'' Then we have that
\begin{equation*}
A=\brac{\omega\in\Omega\colon\omega =\omega _0\omega _1\cdots ,\ \omega _i=0\hbox{ for some }i=0,1,\ldots}
\end{equation*}
Define $B=\brac{1,2,\ldots ,m-1}$, $B_1=B\times S\times S\times\cdots$,
$B_2=B\times B\times S\times S\times\cdots$. Then $B_i\in\cscript$, $B_1\supseteq B_2\supseteq\cdots$ and
$A'=\cap B_i\in\ascript$ is the event ``the system never visits the origin.'' Now for $i=1,2,\ldots$, we have
\begin{equation*}
\nu (B_i)=\paren{\frac{m-1}{m}}^i
\end{equation*}
Hence,
\begin{equation*}
\nu (A')=\nu (\cap B_i)=\lim\nu (B_i)=\lim\paren{\frac{m-1}{m}}^i=0
\end{equation*}
Hence, $\nu (A)=1$. We will show later that the $q$-measure of $A$ is also 1.

Let $U(s,r)$, $r\le s\in\positive$ be a finite unitary system on $\cscript ^m$ and as in Section~4, let
$H_n=(\cscript ^m)^{\otimes (n+1)}$ be the $n$-path space. According to Theorem~\ref{thm44}, the operators
$\dhat _n$ are states on $H_n$, $n=1,2,\ldots$, where we identify $H_n$ with
$\brac{\phi\colon\Omega _n\to\complex}$. Then $\brac{\chi _{\brac{\gamma}}\colon\gamma\in\Omega _n}$ becomes an orthonormal basis for $H_n$. For $\gamma =\gamma _0\gamma _1\cdots\gamma _n\in\Omega _n$ define the cylinder set $\rmcyl (\gamma )$ as the subset of $\Omega$ given by
\begin{equation*}
\rmcyl (\gamma )
  =\brac{\gamma _0}\times\brac{\gamma _1}\times\cdots\times\brac{\gamma _n}\times S\times S\times\cdots
\end{equation*}
Then
\begin{equation*}
\gammahat =m^{(n+1)/2}\chi _{\rmcyl (\gamma )}
\end{equation*}
is a unit vector in $H$ and we define $U_n\chi _{\brac{\gamma}}=\gammahat$. Extending $U_n$ by linearity,
$U_n\colon H_n\to H$ becomes a unitary operator from $H_n$ into $H$. Letting $P_n$ be the projection of $H$ onto the subspace $U_nH_n$ we have
\begin{equation*}
P_nf=\sum _{\gammahat}\elbows{\gammahat ,f}\gammahat
  =m^{(n+1)}\sum _{\gamma\in\Omega _n}\int f\chi _{\rmcyl (\gamma )}d\nu\chi _{\rmcyl (\gamma )}
\end{equation*}
In particular, for $A\in\ascript$ we have
\begin{equation*}
P_n\chi _A=m^{(n+1)}\sum _{\gamma\in\Omega _n}\nu\sqbrac{A\cap\rmcyl (\gamma )}\chi _{\rmcyl (\gamma )}
\end{equation*}
Hence,
\begin{equation*}
P_n1=\sum _{\gamma\in\Omega _n}\chi _{\rmcyl (\gamma )}=1
\end{equation*}
It is also clear that $\rho _n=U_n\dhat _nU_n^*P_n$ is a state on $H$ and also on $U_nH_n$.

Let $\ascript _t$ be the algebra of all time-$t$ cylinder sets
\begin{equation*}
A=A_0\times\cdots\times A_t\times S\times S\times\cdots
\end{equation*}
where $A_i\subseteq S$, $i=0,1,\ldots ,t$. Then $\ascript _t\subseteq\ascript$, $t=0,1,\ldots$, is an increasing sequence of $\sigma$-algebras generating $\ascript$. As in Section~3, let $\nu _t$ be the restriction of $\nu$ to
$\ascript _t$, $t=0,1,\ldots\,$. Then $U_tH_t$ is isomorphic to $L_2(\Omega ,\ascript _t,\nu _t)$, $t=0,1,\ldots$, and forms an increasing sequence of subspaces on $H$.

\begin{thm}       
\label{thm51}
The sequence of states $\rho _t$, $t=0,1,\ldots$, is consistent.
\end{thm}
\begin{proof}
Let $D_n$ be the decoherence functional on $H_n$ given by $D_{\rho _n}$. To show that $\rho _t$ is consistent we must show that
\begin{equation}         
\label{eq51}
D_{n+1}(A\times S,B\times S)=D_n(A,B)
\end{equation}
for every $A,B\subseteq\Omega _n$. Using the notation $\gamma j=\gamma _0\gamma _1\cdots\gamma _nj$,
\eqref{eq51} is equivalent to
\begin{equation}         
\label{eq52}
D_n(\gamma ,\gamma ')=\sum _{j,k=0}^{m-1}D_{n+1}(\gamma j,\gamma 'k)
  =\sum _{j=0}^{m-1}D_n(\gamma j,\gamma 'j)
\end{equation}
for all $\gamma ,\gamma '\in\Omega _n$. Since
\begin{equation*}
\sum _j\elbows{U(n+1,n)e_{\gamma '_n},e_j}\elbows{e_n,U(n+1, n)e_{\gamma _n}}
  =\delta _{\gamma _n,\gamma '_n}
\end{equation*}
it follows that
\begin{equation*}
\sum _jb(\gamma j)\overline{b(\gamma 'j)}=b(\gamma )\overline{b(\gamma ')}\delta _{\gamma _n,\gamma '_n}
\end{equation*}
for all $\gamma ,\gamma '\in\Omega _n$. Hence,
\begin{align*}
\sum _jD_{n+1}(\gamma j,\gamma 'j)&=\sum _ja_\psi (\gamma j)\overline{a_\psi (\gamma 'j)}
  =\sum _jb(\gamma j)\overline{b(\gamma 'j)}\psi (\gamma _0)\overline{\psi (\gamma '_0)}\\
  &=b(\gamma )\overline{b(\gamma ')}\psi (\gamma _0)\overline{\psi (\gamma '_0)}
  =a_\psi\overline{a_\psi (\gamma ')}\delta _{\gamma _n,\gamma '_n}\\
  &=D_n(\gamma ,\gamma ')
\end{align*}
so \eqref{eq52} holds.
\end{proof}

We call the consistent sequence $\rho _t$, $t=0,1,\ldots$, a \textit{finite unitary process}. It follows from
Theorem~\ref{thm31} that the quadratic algebra $\sscript$ of suitable sets contains $\cscript$ and $\mutilde$ is a
$q$-measure on $\sscript$ that extends the natural $q$-measure $\mu$ on $\cscript$. We now consider some sets in $\sscript\smallsetminus\cscript$.

If $\gamma =\gamma _0\gamma _1\cdots\in\Omega$, letting $A_n=\rmcyl (\gamma _0\gamma _1\cdots\gamma _n)$ we have that $\brac{\gamma}=\cap A_n$. Hence,
\begin{equation*}
\nu\paren{\brac{\gamma}}=\lim _{n\to\infty}\nu (A_n)=\lim _{n\to\infty}\,\frac{1}{m^{(n+1)}}=0
\end{equation*}
It follows that $\chi _{\brac{\gamma}}=0$ as an element of $H$ so
\begin{equation*}
\lim\elbows{e_n\chi _{\brac{\gamma}},\chi _{\brac{\gamma}}}=0
\end{equation*}
We conclude that $\brac{\gamma}\in\sscript$ and $\mutilde\paren{\brac{\gamma}}=0$. Since
$\brac{\gamma}\notin\cscript$ we have that $\sscript$ is a proper extension of $\cscript$. More generally, if $B$ is a countable subset of $\Omega$, then $B\in\ascript$ and $\nu (B)=0$. Hence, $B\in\sscript$ and $\mutilde (B)=0$. Moreover, since $\nu (B')=1$, $\chi _{B'}=1\hbox{ a.e.}[\nu ]$. Since $\mu _n(\Omega _n)=1$ we have that
\begin{align*}
\elbows{e_n\chi _{B'},\chi _{B'}}&=\elbows{e_n1,1}=\elbows{U_n\dhat _nU_n^*P_n1,1}\\
  &=\elbows{\dhat _nU_n^*1,U_n1}=\elbows{\dhat _n\chi _{\Omega _n},\chi _{\Omega _n}}\\
  &=\mu _n(\Omega _n)=1
\end{align*}
Hence, $B'\in\sscript$ and $\mutilde (B')=1$.

Let $A\in\ascript$ be the event ``the system visits the origin'' of Example~1. Then $A'$ is the event ``the system never visits the origin'' and we showed in Example~1 that $\nu (A')=0$. Hence $A'\in\sscript\smallsetminus\cscript$ and
$\mutilde (A')=0$. Thus, the $q$-propensity that the system never visits the origin is 0. Since $\nu (A)=1$,
$\chi _A=1\hbox{ a.e.}[\nu ]$. As before we have that $A\in\sscript$ and $\mutilde (A)=1$. Of course, the origin can be replaced by any of the sites $1,2,\ldots ,m-1$. 

As another example, let $B_t$ be the event ``the system visits site $j$ for the first time at $t$.'' Letting
\begin{equation*}
C=\brac{0,1,\ldots ,j-1,j+1,\ldots ,m-1}
\end{equation*}
we see that $B_t\in\cscript$ and
\begin{equation*}
B_t=C\times C\times\cdots\times C\times\brac{j}\times S\times S\times\cdots
\end{equation*}
where there are $t$ factors of $C$. We have that
\begin{equation*}
\nu (B_t)=\paren{\frac{m-1}{m}}^t\frac{1}{m}=\frac{(m-1)^t}{m^{t+1}}
\end{equation*}
Since $B_t\in\cscript$ we have that
\begin{equation*}
\mutilde (B_t)=\mu _t(B_t)=\elbows{\dhat _t\chi _{B_t},\chi _{B_t}}
\end{equation*}
which depends on $U$.

We close this section with an example of a two-site quantum random walk or a two-hopper \cite{gsapp, sor11}. The unitary operator
\begin{equation*}
U=\frac{1}{\sqrt{2\,}}
\left[\begin{matrix}\noalign{\smallskip}1&i\\\noalign{\smallskip}i&1\\\noalign{\smallskip}\end{matrix}\right]
\end{equation*}
generates a stationary unitary system on $\complex ^2$. We think of this system as a particle that is located at one of the two-sites $S=\brac{0,1}$. We assume that the initial state is $e_0$ so the particle always begins at site 0. In this case, we can let $\Omega =\brac{0}\times S\times S\times\cdots$ and we obtain a finite unitary process on
$H=L_2(\Omega ,\ascript ,\nu )$. Let $\mu _n$ be the $q$-measure on $\ascript _n=2^{\Omega _n}$ given by
$\mu _n(A)=\elbows{\dhat _n\chi _A,\chi _A}$. Let $C_t$ be the event ``the particle is at site 1 at time $t$.'' Then
$C_t\in\cscript$ and we have
\begin{equation*}
C_t=\brac{0}\times S\times\cdots\times S\times\brac{1}\times S\times S\cdots
\end{equation*}
Letting $E_t=\brac{0}\times S\times\cdots\times S\times\brac{1}$ we have that $E_t\in\ascript _t$ and
$\mutilde (C_t)=\mu _t(E_t)$. We now compute $\mu _t(E_t)$ for $t=1,2,\ldots\,$.

Of course, $\nu _t(E_t)=1/2$ for $t=1,2,\ldots$, which is the usual classical result. As we shall see, the quantum result is somewhat surprising. This result suggests a periodic motion with period 4. To compute $\mu _t(E_t)$ we shall employ Corollary~\ref{cor47}. We see from the form of $U$ that $\mu _n\paren{\brac{0}}=1/2^n$ for every
$\gamma\in\Omega _n$. Hence, by Theorem~\ref{thm46}, $\dhat _n$ has eigenvalues $1/2$ with multiplicity 2 and eigenvalues 0 with multiplicity $2^n-2$. For $\gamma\in\Omega _n$, $\gamma$ is the binary representation for a unique integer
$a\in\brac{0,1,\ldots ,2^n-1}$. We then identify $\Omega _n$ with $\brac{0,1,\ldots ,2^n-1}$. Let $c(\gamma )$ be the number of position changes (bit flips) in $\gamma$. Applying Theorem~4.7 of \cite{gsapp}, the unit eigenvectors corresponding to $1/2$ are $\psi _0^n,\psi _1^n$ given by
\begin{equation*}
\psi _0^n=\frac{1}{2^{(n-1)/2}}
\left[\begin{matrix}\noalign{\smallskip}i^{c_n(0)}\\0\\i^{c_n(2)}\\0\\\vdots\\i^{c_n(2^n-2)}\\0\\\end{matrix}\right],\quad
\psi _1^n=\frac{1}{2^{(n-1)/2}}
\left[\begin{matrix}\noalign{\smallskip}0\\i^{c_n(1)}\\0\\i^{c_n(3)}\\0\\\vdots\\0\\i^{c_n(2^n-1)}\\\end{matrix}\right]
\end{equation*}
In order to compute $\psi _0^n,\psi _1^n$ the following lemma is useful.

\begin{lem}       
\label{lem52}
{\rm\cite{gsapp}}. For $n\in\positive$, $j=0,1,\ldots ,2^n-1$, the function $c_n(j)$ satisfies
\begin{equation*}
c_{n+1}(2^{n+1}-1-j)=c_n(j)+1
\end{equation*}
\end{lem}

We employ the vector notation $c_n=\paren{c_n(0),c_n(1),\ldots ,c_n(2^n-1)}$. Applying Lemma~\ref{lem52}, since $c_0=(0)$ it follows that $c_1(0,1)$, $c_2=(0,1,2,1)$ and $c_3=(0,1,2,1,2,3,2,1)$. Hence,
\begin{align*}
\psi _0^1&=
\left[\begin{matrix}1\\0\end{matrix}\right],\quad
\psi _1^1=\left[\begin{matrix}0\\i\end{matrix}\right]\\
\psi _0^2&=\frac{1}{\sqrt{2\,}}
\left[\begin{matrix}\phantom{-}1\\\phantom{-}0\\-1\\\phantom{-}0\end{matrix}\right],\quad
\psi _1^2=\frac{1}{\sqrt{2\,}}
\left[\begin{matrix}0\\i\\0\\i\end{matrix}\right]\displaybreak[0]\\
\psi _0^3&=\frac{1}{2}
\left[\begin{matrix}\phantom{-}1\\\phantom{-}0\\-1\\\phantom{-}0\\-1\\\phantom{-}0\\-1\\
  \phantom{-}0\end{matrix}\right],\quad
\psi _0^3=\frac{1}{2}
\left[\begin{matrix}\phantom{-}0\\\phantom{-}i\\\phantom{-}0\\\phantom{-}i\\\phantom{-}0\\-i\\
  \phantom{-}0\\\phantom{-}i\end{matrix}\right],\quad
\end{align*}
Applying Corollary~\ref{cor47} we have that $\mu _0(E_0)=0$ and
\begin{align*}
\mu _1(E_1)&=\mu _1\paren{\brac{1}}=\tfrac{1}{2}\ab{\elbows{\chi _{\brac{1}},\psi _0^1}}^2
  +\tfrac{1}{2}\ab{\elbows{\chi _{\brac{1}},\psi _1^1}}^2\\
  &=1/2\\
  \mu _2(E_2)&=\mu _2\paren{\brac{1,3}}=\tfrac{1}{2}\ab{\elbows{\chi _{\brac{1}},\psi _0^2}
  +\elbows{\chi _{\brac{3}},\psi _0^2}}^2\\
  &\quad +\tfrac{1}{2}\ab{\elbows{\chi _{\brac{1}},\psi _1^2}+\elbows{\chi _{\brac{3}},\psi _1^2}}^2\\
  &=\frac{1}{2}\ab{\frac{1}{\sqrt{2\,}}\,i+\frac{1}{\sqrt{2\,}}\,i}^2=1\\
  \mu _3(E_3)&=\mu _3\paren{\brac{1,3,5,7}}
  =\tfrac{1}{2}\ab{\tfrac{1}{2}\,i+\tfrac{1}{2}\,i-\tfrac{1}{2}\,i+\tfrac{1}{2}\,i}^2=1/2
\end{align*}

In general, we have
\begin{align*}
\mu _t(E_t)&=\mu _t\paren{\brac{1,3,\ldots ,2^t-1}}
  =\frac{1}{2}\ab{\sum _{j=1}^{2^t-1}\elbows{\chi _{\brac{j}},\psi _1^t}}^2\\
  &=\frac{1}{2^t}\ab{\sum _{\substack{j=1\\j\ \mathrm{odd}}}^{2^t-1}i^{c_t(j)}}^2
\end{align*}
Letting
\begin{equation*}
F(t)=\sum _{\substack{j=0\\j\ \mathrm{even}}}^{2^t-2}i^{c_t(j)},\qquad
 G(t)\sum _{\substack{j=1\\j\ \mathrm{odd}}}^{2^t-1}i^{c_t(j)}
\end{equation*}
we have
\begin{equation}         
\label{eq53}
\mu _t(E_t)=\tfrac{1}{2^t}\ab{G(t)}^2
\end{equation}
Notice we have shown that $G(0)=0$, $G(1)=i$, $G(2)=G(3)=2i$.

\begin{lem}       
\label{lem53}
For $j=0,1,2,3$ and $m\in\positive$ we have $G(4m+j)=(-4)^mG(j)$.
\end{lem}
\begin{proof}
Applying Lemma~\ref{lem52} we have that $G(t)=G(t-1)+iF(t-1)$ and $F(t)=F(t-1)+iG(t-1)$. Hence,
\begin{align*}
G(t)&=G(t-1)+i\sqbrac{F(t-2)+iG(t-2)}=G(t-1)-G(t-2)+iF(t-\!2)\\
  &=G(t-1)-G(t-2)+i\sqbrac{F(t-3)+iG(t-3)}\\
  &=G(t-1)-G(t-2)-G(t-3)+iF(t-3)
\end{align*}
Continuing this process we have
\begin{align*}
G(t)&=G(t-1)-G(t-2)-\cdots -G(1)+iF(1)\\
  &=G(t-1)-G(t-2)-\cdots -G(1)+i
\end{align*}
Hence,
\begin{align*}
G(t)&=G(t-2)-G(t-3)-\cdots -G(1)+i\\
  &\quad -G(t-2)-G(t-3)-\cdots -G(i)+i\\
  &=-2\sqbrac{G(t-3)-G(t-4)-\cdots -G(1)+i}\\
  &=-2\sqbrac{G(t-4)-G(t-5)-\cdots -G(1)+i}\\
  &\quad-2\sqbrac{G(t-4)+G(t-5)+\cdots +G(1)-i}\\
  &=-4G(t-4)
\end{align*}
Letting $t=4m+j$, $j=0,1,2,3$ we have
\begin{align*}
G(t)&=-4G(4(m-1)+j)=-4\sqbrac{-4G(4(m-2)+j)}\\
  &=(-4)^2G(4(m-2)+j)
\end{align*}
Continuing this process gives our result
\end{proof}

Applying \eqref{eq53} and Lemma~\ref{lem53} we have for $t=4m+j$, $j=0,1,2,3$ and $m\in\positive$ that
\begin{equation*}
\mu _t(E_t)=\frac{1}{2^{4m+j}}\,4^{2m}\ab{G(j)}^2=\frac{1}{2^j}\,\ab{G(j)}^2
\end{equation*}
This shows that $\mu _t(E_t)$ is periodic with period 4. The first few values of $\mu _t(E_t)$ are $0,1/2,1,1/2,0,1/2,1,1/2,0,\ldots\,$.

Let $F_t$ be the event ``the particle is at site 0 at time $t$.'' Then $F_t\in\cscript$ and we have
\begin{equation*}
F_t=\brac{0}\times S\times\cdots\times S\times\brac{0}\times S\times S\cdots
\end{equation*}
Again, setting
\begin{equation*}
G_t=\brac{0}\times S\times\cdots\times S\times\brac{0}
\end{equation*}
we have that $G_t\in\ascript _t$ and $\mutilde (F_t)=\mu _t(G_t)$. Of course, $G_t=E'_t$ and we have computed
$\mu _t(E_t)$. It does not immediately follow that $\mu _t(G_t)=1-\mu _t(E_t)$ because $\mu _t$ is not additive. However, as in \eqref{eq53} we have that $\mu _t(G_t)=2^{-t}\ab{F(t)}^2$ where $F(0)=1$, $F(1)=1$, $F(2)=0$,
$F(3)=-2$. As in Lemma~\ref{lem53}, for $j=0,1,2,3$ and $m\in\positive$ we have that $F(4m+j)=(-4)^mF(j)$. We conclude that in this case, $\mu _t(G_t)=1-\mu _t(E_t)$. Thus, $\mu _t(G_t)$ is periodic with period 4. The first few values of $\mu _t (G_t)$ are $1,1/2,0,1/2,1,1/2,0,1/2,\cdots\,$.

Finally, let $f_t$ be the random variable that gives the position of the particle at time $t$. Thus, for
$\gamma\in\Omega$, $\gamma =\gamma _0\gamma _1\cdots$, $f_t(\gamma )=\gamma _t$. Then $f=\chi _{C_t}$
and we have
\begin{equation*}
\int f_td\mutilde =\mutilde (C_t)=\mu _t(E_t)=\frac{1}{2^t}\,\ab{G(t)}^2
\end{equation*}
which we have already computed.

\section{Quantum Integrals} 
Let $\rho _t$, $t=0,1,2,\ldots$, be a discrete $q$-process on $H=L_2(\Omega ,\ascript ,\nu)$. As in Section~3, a random variable $f\in H$ is integrable if $\lim (\rho _t\fhat )$ exists and is finite, in which case $\int fd\mutilde$ is this limit. If a global state $\rho$ exists, it is also of interest to compute the integral $\int fd\mu _\rho =\rmtr (\rho\fhat )$. The operator $\fhat$ can be complicated and the expression $\rmtr (\rho\fhat )$ difficult to evaluate. This section considers the case in which $f$ is a simple function. A general random variable can be treated using Corollary~\ref{cor26}(d). We first find the eigenvectors and eigenvalues for two-valued random variables.

If $f$ has the form $f=\alpha\chi _A$, $\alpha\in\real$, then $\fhat =\alpha\chihat _A$ and
\begin{equation*}
\int fd\mu _t=\rmtr (\rho\fhat\,)=\alpha\elbows{\rho\chi _A,\chi _A}
\end{equation*}
The next result treats nonnegative random variables with two nonzero values

\begin{thm}       
\label{thm61}
If $f=\alpha\chi _A+\gamma\chi _B$ where $A\cap B=\emptyset$, $\nu (A)\nu (B)\ne 0$ and $0<\alpha <\beta$ then
$\fhat$ has two nonzero eigenvalues
\begin{equation*}
\lambda _{\pm}
  =\frac{\alpha\nu (A)+\gamma\nu (B)\pm\sqrt{\sqbrac{\alpha\nu (A)-\gamma\nu (B)}^2+2\alpha ^2\nu (A)\nu (B)}}
  {2\nu (B)}
\end{equation*}
with corresponding eigenvectors
\begin{equation*}
g_{\pm}=\chi _A+b_{\pm}\chi _B
\end{equation*}
where $b_{\pm}=\sqbrac{\lambda _{\pm}-\nu (A)}/\nu (B)$.
\end{thm}
\begin{proof}
We first treat the case in which $f$ has the form
\begin{equation*}
f=\chi _A+\beta\chi _B
\end{equation*}
where $A\cap B=\emptyset$, $\nu (A)\nu (B)\ne 0$ and $\beta >1$. For $g\in H$ we have
\begin{align}         
\label{eq61}
\fhat g(x)&=\int\min\paren{f(x),f(y)}g(y)d\nu (y)\notag\\
  &=\int _{\brac{y\colon f(y)\ge f(x)}}f(x)g(y)d\nu (y)+\int _{\brac{y\colon f(y)<f(x)}}f(y)g(y)d\nu (y)
\end{align}
Letting $C=(A\cup B)'=A'\cap B'$ it follows from \eqref{eq61} that:

If $x\in C$, then $(\fhat g)(x)=0$

If $x\in A$, then $(\fhat g)(x)=\int _{A\cup B}g(y)d\nu (y)$

If $x\in B$ then $(\fhat g)(x)=\beta\int _Bg(x)d\nu (y)+\int _Ag(y)d\nu (y)$\newline
We conclude that
\begin{equation}         
\label{eq62}
(\fhat g)(x)=\int _{A\cup B}g(y)d\nu (y)\chi _A+\sqbrac{\beta\int _Bg(y)d\nu (y)+\int _Ag(y)d\nu (y)}\chi _B
\end{equation}
If $g\perp\chi _A$ and $g\perp\chi _B$ then $\fhat g=0$. Thus,
\begin{equation*}
\rmrange (\fhat\,)=\rmspan\brac{\chi _A,\chi _B}
\end{equation*}
It follows that eigenvectors corresponding to nonzero eigenvalues have the form $g=a\chi _A+b\chi _B$,
$a,b\in\complex$. Assuming that $g$ has this form and $\fhat g=\lambda g$, applying \eqref{eq62} gives
\begin{align*}
(\fhat g)(x)&=\sqbrac{a\nu (A)+b\nu (B)}\chi _A+\sqbrac{\beta b\nu (B)+a\nu (A)}\chi _B\\
  &=a\lambda\chi _A+b\lambda\chi _B
\end{align*}
Hence,
\begin{equation*}
a\lambda =a\nu (A)+b\nu (B),\quad b\lambda =\beta b\nu (B)+a\nu (A)
\end{equation*}
Notice that $a\ne 0$ since otherwise $b=0$ which contradicts $g\ne 0$. We can therefore assume that $a=1$. Then
$b=\sqbrac{\lambda -\nu (A)}/\nu (B)$ and $b\sqbrac{\lambda -\beta\nu (B)}=\nu (A)$. Eliminating $b$ from these two equations gives
\begin{equation*}
\lambda ^2-\sqbrac{\nu (A)+\beta\nu (B)}\lambda +(\beta -1)\nu (A)\nu (B)=0
\end{equation*}
Applying the quadratic formula we have the two solutions
\begin{equation*}
\lambda _{\pm}=\frac{\nu (A)+\beta\nu (B)\pm\sqrt{\sqbrac{\nu (A)-\beta\nu (B)}^2+2\nu A()\nu (B)}}{2\nu (B)}
\end{equation*}
and corresponding eigenvectors are
\begin{equation*}
g_{\pm}=\chi _A+b_{\pm}\chi _B
\end{equation*}
where $b_{\pm}=\frac{\lambda _{\pm}-\nu (A)}{\nu (B)}$. Since the original $f$ satisfies
\begin{equation*}
f=\alpha\sqbrac{\chi _A+\frac{\gamma}{\alpha}\chi _B}
\end{equation*}
letting $\beta =\gamma /\alpha$ and multiplying $\lambda _{\pm}$ by $\alpha$ gives the result.
\end{proof}

The eigenvectors $g_{\pm}$ need not be normalized but this can easily be done. Theorem~\ref{thm61} treats random variables with two positive values. The remaining case for two-valued random variables is one of the form
$f=\alpha\chi _A+\beta\chi _B$ where $A\cap B=\emptyset$, $\nu (A)\nu (B)\ne 0$ and $\alpha >0$, $\beta <0$. This case is easy to treat because we can write $f=\alpha\chi _A-(-\beta)\chi _B$ so that
\begin{equation*}
\fhat =\alpha\chihat _A-(-\beta )\chihat _B=\alpha\ket{\chi _A}\bra{\chi _A}+\beta\ket{\chi _B}\bra{\chi _B}
\end{equation*}
Hence, the nonzero eigenvalues of $\fhat$ are $\alpha\nu (A)$, $\beta\nu (B)$ with corresponding unit eigenvectors $\nu (A)^{-1/2}\chi _A$, $\nu (B)^{-1/2}\chi _B$. We can thus find the eigenvalues $\lambda _i$ and normalized eigenvectors $v_i$, $i=1,2$, for an arbitrary two-valued random variable $f$. If $\rho$ is a state and $\mu _\rho$ the corresponding $q$-measure, we have
\begin{align*}
\int fd\mu _\rho&=\rmtr (\rho\fhat\,)=\elbows{\rho\fhat v_1,v_1}+\elbows{\rho\fhat v_2,v_2}\\
  &=\lambda _1\elbows{\rho v_1,v_1}+\lambda _2\elbows{\rho v_2,v_2}
\end{align*}

Now suppose $f=\alpha\chi _A+\beta\chi _B+\gamma\chi _C$ is a three-valued random variable where
$\alpha ,\beta ,\gamma\in\real$ are distinct and $A,B,C$ are mutually disjoint. By Theorem~\ref{thm25}(e) we have
\begin{align}         
\label{eq63}
\fhat&=(\alpha\chi _A+\beta\chi _B+\gamma\chi _C)^\wedge\notag\\
  &=(\alpha\chi _A+\beta\chi _B)^\wedge +(\alpha\chi _A+\gamma\chi _C)^\wedge
  +(\beta\chi _B+\gamma\chi _C)^\wedge -\alpha\chihat _A-\beta\chihat _B-\gamma\chihat _C
\end{align}
The right side of \eqref{eq63} contains the quantization operators for three two-valued random variables. Letting
$\lambda _1^i,\lambda _2^i$ be the eigenvalues of the $i$th operator with corresponding unit eigenvectors
$v_1^i,v_2^i$, $i=1,2,3$ we have
\begin{align*}
\int fd\mu _\rho&=\rmtr (\rho\fhat\,)\\
  &=\sum _{i=1}^3\sum _{j=1}^2\lambda _j^i\elbows{\rho v_j^i,v _j^i}-\alpha\elbows{\rho\chi _A,\chi _A}
  -\beta\elbows{\chi _B,\chi _B}-\gamma\elbows{\rho\chi _C,\chi _C}
\end{align*}
Continuing by induction we have for $f=\sum _{i=1}^n\alpha _i\chi _{A_i}$ that
\begin{equation*}
\fhat =\sum _{i<j=1}^n(\alpha _i\chi _{A_i}+\alpha _j\chi _{A_j})^\wedge-(n-1)\sum _{i=1}^n\alpha _i\chihat _{A_i}
\end{equation*}
Using a similar notation as before gives
\begin{align*}
\int fd\mu _\rho&=\rmtr (\rho\fhat\,)\\
  &=\sum _{i=1}^{n(n-1)/2}\sum _{j=1}^2\lambda _j^i\elbows{\rho v_j^i,v_j^i}
  -(n-1)\sum _{i=1}^n\alpha _i\elbows{\rho\chi _{A_i}\chi _{A_i}}
\end{align*}


\begin{thebibliography}{99}
\bibitem{djs10}F~Dowker, S.~Johnston and R.~Sorkin, Hilbert spaces from path integrals, arXiv: quant-ph
(1002:0589), 2010.
\bibitem{djsu10}F~Dowker, S.~Johnston and S.~Surya, On extending the quantum measure, arXiv: quant-ph (1002:2725), 2010.
\bibitem{dgt08}F~Dowker and Y.~Ghazi-Tabatabai, 
Dynamical wave function collapse models in quantum measure theory, \textit{J. Phys. A} \textbf{41} (2008), 105301.
\bibitem{gt09}Y.~Ghazi-Tabatabai, 
Quantum measure: a new interpretation, arXiv: quant-ph (0906:0294) 2009.
\bibitem{gud09}S.~Gudder, Quantum measure and integration theory, \textit{J. Math. Phys.} \textbf{50},
(2009), 123509.
\bibitem{gud10}S.~Gudder, Quantum measure theory, \textit{Math. Slovaca} \textbf{60}, (2010), 681--700.
\bibitem{gsapp}S.~Gudder and R.~Sorkin, Two-Site quantum random walk, \textit{Gen. Rel. Grav.} (to appear), arXiv: quant-ph (1105.0705), 2011.
\bibitem{gudapp}S.~Gudder, Quantum measures and integrals arXiv: quant-ph (1105.3781), 2011.
\bibitem{hal09}J.~J.~Halliwell, Partial decoherence of histories and the Diosi test, arXiv: quant-ph
(0904.4388), 2009.
\bibitem{mocs05}X.~Martin, D.~O'Connor and R.~Sorkin, Random walk in generalized quantum theory,
\textit{Physic Rev D} \textbf{71}, 024029 (2005).
\bibitem{sor94}R.~Sorkin, 
Quantum mechanics as quantum measure theory, \textit{Mod. Phys. Letts.~A} \textbf{9} (1994), 3119--3127.
\bibitem{sor07}R.~Sorkin, 
Quantum mechanics without the wave function, \textit{Mod. Phys. Letts.~A} \textbf{40} (2007), 3207--3231.
\bibitem{sor11}R.~Sorkin, 
Toward a fundamental theorem of quantum measure theory, arXiv: quant-ph (1104:0997), 2011.
\bibitem{sor10}R.~Sorkin, 
Logic is to the quantum as geometry is to gravity, arXiv: quant-ph (1004.1226) 2010.

\end{thebibliography}
\end{document}